\documentclass[12pt]{article}
\usepackage{mathtools}
\usepackage{amsthm}
\usepackage{amssymb}
\usepackage{subcaption}
\usepackage{comment}
\usepackage{url}
\usepackage{multicol}
\usepackage{tikz}
\usepackage[a4paper, margin=1in]{geometry} 

\newtheorem{theorem}{Theorem}[section]

\newtheorem{conjecture}[theorem]{Conjecture}

\theoremstyle{definition}
\newtheorem{definition}[theorem]{Definition}

\begin{document}
	\title{Predator-dependent replicator dynamics or a predator-prey model with two prey types and frequency dependence}
	\author{Henrique M. Cruz$^{1}$\\ Armando G. M. Neves$^{2}$
	\\
	\normalsize{$^{1}$School of Engineering, Universidade Federal de Minas Gerais}
	\\ 	\normalsize{hmarquescruz@ufmg.br}\\
	\normalsize{$^{2}$Department of Mathematics, Universidade Federal de Minas Gerais}
		\\ 	\normalsize{aneves@mat.ufmg.br}\\
}

\maketitle
	
	\begin{abstract}
		Braga and Wardil [J. Phys. A: Math. Theor. 55 (2022) 025601] introduced a population model for two prey species that compete among themselves and are preyed upon by a single predator species. They showed the existence of 16 dynamic scenarios and stated sufficient conditions for the stable coexistence of the three species. The model, which can be seen as replicator dynamics with predator-dependent fitnesses for the prey, is based on two pay-off matrices: one for prey reproduction and one for interaction between prey and predators.
		
		We argue that the model can also be seen as a Lotka-Volterra-type predator-prey model with a single prey species, logistic limitation for prey, and frequency-dependent reproduction and capture coefficients. Using this alternative viewpoint, we obtain conditions for the existence of equilibria with the three types of individuals. We also prove theorems on the stability or instability of equilibria with only two species and relate the stability change of these equilibria to the appearance or disappearance of equilibria with the three species.
		
		When all parameters, except the one that regulates carrying capacities, are fixed, a rich cascade of bifurcations may appear. Solutions range from predator extinction due to insufficient prey to predators coexisting with one or two prey types. Sometimes stable limit cycles involving all species appear.
		
	\end{abstract}
	
	 
\section{Introduction}
Peacocks are notable for the colorful tails of the male individuals. A beautiful tail may increase the success of an individual in finding a female for mating, but it may also attract predators. How would it be if an alternative peacock type with less colorful tails, thus less attractive to females and predators, competed with the usual peacocks? Would it be possible for the two types to coexist with predators? Or would a single type (which one?) ultimately survive? With peacocks in mind as their primary motivation, but also describing many other possibilities, Braga and Wardil \cite{bw} introduced a deterministic population model for two species (or types) of prey -- let us call them types A and B -- and a single predator species that feeds on both prey species. Prey compete among themselves for limited environmental resources, may have larger or smaller reproductive success, and may be killed more or less easily by predators.

From a purely mathematical point of view, the Braga-Wardil model consists of a system of three ordinary differential equations (ODEs). As it turns out, systems of this dimension are considerably more difficult to analyze than two-dimensional ones. One reason is the absence of an analog of the Poincaré-Bendixson theorem, an important technical tool for the qualitative analysis of two-ODE systems; see, e.g., \cite{hirschsmaledevaney}. Despite that, there is a large literature on three-ODE models of biological systems. For instance, M. L. Zeeman \cite{zeeman} studied competitive systems for three species. Freedman and Waltman's work \cite{frewal} considered systems with three species too, prey and predators, including two prey and one predator species, as in this paper. Other studies of such systems are, e.g., \cite{adamu, koro, sharma, tripathi}.

As introduced by Braga and Wardil \cite{bw}, the model involves two pay-off matrices, one for the reproductive success of the prey types and one for their interaction with predators. If predators are absent, the model becomes the standard replicator equation \cite{tayjon, nowak} for the prey. 
In this case, depending on the entries in the reproduction pay-off matrix, we may have 4 possible dynamics: dominance of either prey type A or prey type B, stable coexistence of both prey types, or codominance \cite{nowak}. If predators are present, the second pay-off matrix comes into play, and 4 independent dynamic possibilities appear. Taking into account the two matrices, we have $4^2=16$ possible scenarios, as shown in Table \ref{tab:reppresce}. Braga and Wardil noticed that the population frequency of prey types obeyed a replicator equation in which the fitness of each type depends on its frequency, as usual in Evolutionary Game Theory \cite{hofsig, nowak}, but also on the size of the predator population. This explains the title of their paper, in which their model is called \emph{predator-dependent replicator dynamics}.

Such a diversity of dynamic scenarios and the fascinating examples of numerical solutions provided in \cite{bw} convinced us that studying the Braga-Wardil (shortened as BW from now on) model further and rigorously would be a delightful task, and could contribute to the dissemination of their model to a wider audience. The present article is the result of our efforts.

Besides predator-dependent replicator dynamics, we show that the BW model can also be seen as a Lotka-Volterra-type predator-prey model for a single prey species with a logistic competition term and coefficients dependent on the population frequency of prey type A. Using this alternative viewpoint, we prove that predators are extinct unless the carrying capacities of the prey are sufficiently large. We study in detail all equilibrium solutions in which only two types of individuals are present and obtain a new equation for calculating equilibrium solutions in which the three types are present (ABY equilibria). We provide several examples regarding the number of ABY equilibria in which we numerically evaluate whether they are stable or unstable. We prove important results on stability or instability of equilibrium solutions with only two species present. Despite that, stability (or not) of ABY equilibria turns out to be more elusive. We approach this question by using methods from Bifurcation Theory \cite{guckhol, perko}.

The paper is organized as follows. After reviewing in Section \ref{sec:rev} the construction of the BW model and some of the results in \cite{bw}, in Section \ref{sec:rew} we rewrite the model as a system in which the first equation is the predator-dependent replicator and the other two equations are the frequency-dependent Lotka-Volterra logistic predator-prey (LVLPP) model. In Section \ref{sec:ABYeq} we give some examples of ABY equilibria and numerically study whether they are stable or unstable. These examples are further explored in later sections of the paper.	In Section \ref{sec:stabAYBYAB} we prove theorems on the stability of equilibria with only two species present, relating them to creation or annihilation of ABY equilibria. These theorems are also interpreted in game-theoretic terms. In Section \ref{sec:bif} we prove that saddle-node bifurcations \cite{guckhol, perko} occur in the model. We also observe that bifurcations similar to transcritical occur and we state a conjecture on such bifurcations. We use the results on bifurcations to assess whether ABY equilibria with the three types are stable or unstable. In Section \ref{sec:hopf} we provide an example of a numerical solution to the model showing that the three types may coexist stably not only at equilibria, but also as limit cycles arising probably due to Hopf bifurcations \cite{guckhol, perko}. The paper closes with a Conclusions section, Section \ref{sec:conc}. Appendix \ref{app:lvlpp} is devoted to a review of the Lotka-Volterra logistic predator-prey model. Appendix \ref{app:sntc} states some results on Bifurcation Theory used in the main text.

\section{Review on the BW model}\label{sec:rev}
The model, as proposed by Braga and Wardil \cite{bw}, with slight changes in notation, is 
\begin{equation} \label{modelorig}
	\begin{cases} 
		A' =  \left( \frac{aA + bB}{A + B} - \frac{A + B}{K} - \frac{a^*A + b^*B}{A + B}Z \right) \,A \\
		B' = \left( \frac{cA + dB}{A + B} - \frac{A + B}{K} - \frac{c^*A + d^*B}{A + B}Z \right)\, B \\
		Z' = \left( -\beta + \gamma \frac{a^*A + b^*B}{A + B}A + \gamma \frac{c^*A + d^*B}{A + B}B \right) \,Z
	\end{cases}\;.
\end{equation}
$A$, $B$ and $Z$ are, respectively, the population sizes of the two prey types and of the predator species. All parameters in the equations are \textit{positive}, and we explain all terms in the following.

The term $\frac{aA + bB}{A + B}$ is a birth rate for the A prey population. If only the As are present, i.e. $B=0$, it reduces to $a$. Thus, $a$ is the birth rate for the A population if only As are present. Similarly, making $A=0$ (or, more precisely, taking $A$ very small), we see that $b$ is the birth rate for the As if only Bs are present. Analogously, $c$ and $d$ are, respectively, the birth rates for the B population if only As are present and only Bs are present. 

Defining 
\begin{equation} \label{defx}
	x= \frac{A}{A+B}
\end{equation}
as the population frequency of the A type in the total prey population, we see that $\frac{aA + bB}{A + B}=ax+b(1-x)$ is a \emph{frequency-dependent} birth rate for the A population, and $\frac{cA + dB}{A + B}=cx+d(1-x)$ is a frequency-dependent birth rate for the B population. 

In order to connect the model with well-known results in Evolutionary Game Theory, we may group parameters $a$, $b$, $c$ and $d$ in a \emph{reproduction} pay-off matrix
\begin{equation}  \label{defrepmat}
	M_R=\begin{pmatrix}
		a & b \\
		c & d
	\end{pmatrix}\;.
\end{equation}

In fact, if predators are absent, $Z=0$, we may use the first two equations in \eqref{modelorig} to show that the prey-population fraction of As obeys
\begin{equation}  \label{standrep}
	x'= (f_A(x)-f_B(x))x(1-x)\;,
\end{equation}
that readers familiar with Evolutionary Game Theory will recognize as the standard \emph{replicator equation} when the population consists of only two types of individuals, where 
\[f_A(x)= a x+b(1-x)\]
and
\[f_B(x)=c x+d(1-x)\]
are the standard fitness functions for A and B individuals, respectively. For the standard replicator equation, there exist four possible evolutionary scenarios \cite{nowak}, displayed in Table \ref{tab:repsce} along with the inequalities characterizing them.
\begin{table}   
	\begin{tabular}{cc}
		\hline\noalign{\smallskip}
		A dominance& $a>c$ and $b>d$   \\
		\noalign{\smallskip}\hline\noalign{\smallskip}
		Codominance & $a>c$ and $b<d$   \\
		\noalign{\smallskip}\hline\noalign{\smallskip}
		Coexistence & $a<c$ and $b>d$   \\
		\noalign{\smallskip}\hline\noalign{\smallskip}
		B dominance & $a<c$ and $b<d$   \\
		\noalign{\smallskip}\hline \smallskip
	\end{tabular}
	\caption{\label{tab:repsce} The 4 possible evolutionary scenarios in the standard replicator equation and the inequalities in the reproduction pay-off matrix characterizing them.}
\end{table}

Before considering the presence of the predators, we recognize the term $\frac{A + B}{K}$ present in the first two equations of \eqref{modelorig} as a competition term representing the fact that all prey compete among themselves -- regardless of type -- for environmental resources.  Parameter $K$ \emph{is not} the carrying capacity for the prey. In fact, it can be seen that, in the absence of predators and B prey, the first of eqs. \eqref{modelorig} becomes a simple \emph{logistic} model \cite{edel, murray} with the A prey equilibrium population equal to $Ka$. Likewise, $Kd$ is the B prey equilibrium population when the other two types are absent. As $K$ is a parameter that controls carrying capacities, we call it the \emph{carrying capacity control} parameter.

If only A prey and predators are present, i.e. $B=0$, \eqref{modelorig} reduces to 
\begin{equation} \label{modlvlpp}
	\begin{cases} 
		A' =  \left( a - \frac{A}{K} - a^* Z \right) \,A \\
		Z' = \left( -\beta + \gamma a^* A  \right) \,Z
	\end{cases}\;,
\end{equation}
which is an instance of what we will refer to as the Lotka-Volterra \emph{logistic} predator-prey (LVLPP) model \cite{hofsig}. In the LVLPP model, parameter $a^*$ is the capture rate. Of course, if only B prey and predators were present, \eqref{modelorig} would also become an LVLPP model for the Bs. In Appendix \ref{app:lvlpp} we review some results on the LVLPP model that will be useful for understanding the BW model.

Returning to \eqref{modelorig}, $\frac{a^*A + b^*B}{A + B}$ may now be interpreted as the frequency-dependent capture rate for the A prey. Similarly, $\frac{c^*A + d^*B}{A + B}$ is the frequency-dependent capture rate for the Bs. $a^*$ is the capture rate for A prey in a population with only A prey, $b^*$ is the capture rate for A prey in a population with only B prey, and so on. As with reproduction parameters, predation parameters $a^*$, $b^*$, $c^*$ and $d^*$ will be conveniently grouped in a predation pay-off matrix
\[\begin{pmatrix}
	a^* & b^* \\
	c^* & d^*
\end{pmatrix}\;.\]
The fact that the capture rates are not necessarily equal means that predators' preference for each kind of prey may depend on the prey population composition.

Finally, making $A=0$ and $B=0$ in \eqref{modelorig}, the model reduces to $Z'=-\beta Z$, which describes exponential decay of predators if there are no prey to feed them. Thus, $\beta$ is interpreted as the decay rate of predators.

Before proceeding, we may take the time to eliminate a non-essential parameter in the model. We rescale the predator population $Z$ by defining a new variable $Y$ as 
\[
Z = \frac{\gamma}{\beta} Y.
\]

Substituting this into the system leads to the transformed equations
\begin{equation} \label{BWmodel}
	\begin{cases} 
		A' =  \left( \frac{aA + bB}{A + B} - \frac{A + B}{K} - \frac{a'A + b'B}{A + B}Y \right)\;A \\
		B' =  \left( \frac{cA + dB}{A + B} - \frac{A + B}{K} - \frac{c'A + d'B}{A + B}Y \right)\;B \\
		Y' = \beta \left( -1 +\frac{a'A + b'B}{A + B}A + \frac{c'A + d'B}{A + B}B \right) \,Y
	\end{cases}\;, 
\end{equation}
where $a'=\frac{\gamma}{\beta}a^*$ and the same rescaling defines $b'$, $c'$ and $d'$. The above equations are what we will refer to hereafter as the BW model. In analogy with \eqref{defrepmat}, matrix $M_P$ below will be called the \emph{predation pay-off matrix}:
\begin{equation}  \label{defpredmat}
	M_P=\begin{pmatrix}
		a' & b' \\
		c' & d'
	\end{pmatrix}\;.
\end{equation}

As before, with $Z=0$, we now consider the presence of predators and combine the first two equations in \eqref{BWmodel} to obtain
\begin{equation}  \label{repl}
	x'= [f_A(x,Y)-f_B(x,Y)] \, x(1-x)\;,
\end{equation}
where 
\begin{equation} \label{A fit}
	f_A(x,Y) = (a-a'Y)x+(b-b'Y)(1-x)
\end{equation}
and  
\begin{equation} \label{B fit}
	f_B(x,Y) = (c-c'Y)x+(d-d'Y)(1-x)\;
\end{equation}
are fitness functions for the two prey types, now dependent also on the predator population size $Y$. Equation \eqref{repl} generalizes \eqref{standrep} and is what Braga and Wardil termed the \emph{predator-dependent} replicator equation.

\begin{table}     
	\begin{tabular}{ccccc}
		\hline\noalign{\smallskip}
		&$\begin{array}{c}
			\textrm{A dominance}\\
			a'<c'\\
			b'<d'
		\end{array}$&$\begin{array}{c}
			\textrm{Codominance}\\
			a'<c'\\
			b'>d'
		\end{array}$ &$\begin{array}{c}
			\textrm{Coexistence}\\
			a'>c'\\
			b'<d'
		\end{array}$ & $\begin{array}{c}
			\textrm{B dominance}\\
			a'>c'\\
			b'>d'
		\end{array}$ \\
		
		\noalign{\smallskip}\hline\noalign{\smallskip}
		$\begin{array}{c}
			\textrm{A dominance}\\
			a>c\\
			b>d
		\end{array}$ &1&2&3&4 \\
		\noalign{\smallskip}\hline
		$\begin{array}{c}
			\textrm{Codominance}\\
			a>c\\
			b<d
		\end{array}$ &5&6&7&8 \\
		\noalign{\smallskip}\hline
		$\begin{array}{c}
			\textrm{Coexistence}\\
			a<c\\
			b>d
		\end{array}$ &9&10&11&12 \\
		\noalign{\smallskip}\hline
		$\begin{array}{c}
			\textrm{B dominance}\\
			a<c\\
			b<d
		\end{array}$ &13&14&15&16 \\
		\noalign{\smallskip}\hline \smallskip
	\end{tabular}
	\caption{\label{tab:reppresce} The 16 possible evolutionary scenarios in the predator-dependent replicator equation obtained by independently combining inequalities in the reproduction and predation pay-off matrices. The numbering of the scenarios is the same as in \cite{bw}.}
\end{table}

We may combine inequalities between elements in the reproduction matrix $M_R$, as in Table \ref{tab:repsce}, with analogous inequalities in the predation matrix $M_P$. The resulting scenarios are shown in Table \ref{tab:reppresce} with the same numbering as they appear in \cite{bw}. As seen in \eqref{A fit} and \eqref{B fit}, the predation pay-offs decrease the corresponding fitness, instead of increasing it. This is why in Table \ref{tab:reppresce} the inequalities between elements of $M_P$ are reversed with respect to the inequalities between elements of $M_R$.

Along with the analog of Table \ref{tab:reppresce}, Braga and Wardil \cite{bw} used a color code to classify the 16 scenarios. Their classification regards whether conditions for the long-term coexistence of the three species cannot be satisfied, or can be satisfied, or are always satisfied. This classification is obtained using a theorem of Freedman and Waltman \cite{frewal}, stating \emph{sufficient conditions} for long-term coexistence of the three species in Appendix A of \cite{bw}. We stress that the classification result in their paper is also based on the approximation that the carrying capacity control parameter $K$ is very large. 

We must point out that one of the hypotheses of the Freedman-Waltman theorem, as stated in \cite{bw, frewal}, is that the functions appearing on the right-hand side of eqs. \eqref{BWmodel} are $C^1$. Of course, the right-hand sides are not even well-defined when $A=B=0$. Despite that, they may be defined as limits when $(A,B) \rightarrow (0,0)$ and become continuous functions, but non-differentiable at $(A,B)=(0,0)$. This fact means that, as stated, the Freedman-Waltman theorem is not applicable to the BW model. It is possible that the hypothesis of continuous differentiability when $(A,B)=(0,0)$ is not necessary, because we have found that the classification in \cite{bw} remains correct. 

But instead of trying to prove that the Freedman-Waltman theorem is applicable, we followed the more constructive path of determining all equilibrium solutions of \eqref{BWmodel}, classifying their stability, and studying the bifurcations to which they are subject. Moreover, we do not restrict our results to the limit of infinite $K$. 

It turns out that the sufficient conditions of the Freedman-Waltman theorem and the large $K$ requirement are too restrictive and we were able to obtain the possibility of long-term coexistence of the three species in all scenarios, except for scenarios 1 and 16. The inequalities that characterize scenarios 1 and 16, used in \eqref{repl}, show that either $x$ is always increasing (scenario 1), or always decreasing (scenario 16). If the solutions are well-defined up to infinite time, either B prey will be extinct, or A prey will be extinct.

\section{Rewriting the model and some consequences} \label{sec:rew}

Although useful in qualitative understanding of the model in terms of Evolutionary Game Theory, eq. \eqref{repl} alone cannot replace the three equations in \eqref{BWmodel}. In order to complete our understanding of the dynamics, we must complement \eqref{repl} with an equation for the size of the \textit{total} prey population and another for the size of the predator population.

We denote the total prey population as
\begin{equation} \label{totalprey}
	P=A+B\;.
\end{equation} 
By summing the first two in \eqref{BWmodel} we obtain, after some simplification, an equation for $P'$. And we may rewrite the third of eqs. \eqref{BWmodel} in terms of $x$, $P$ and $Y$. The result is
\begin{equation}  \label{xpyBW}
	\begin{cases} 
		x' = x(1-x)(f_A(x,Y)-f_B(x,Y))\\
		P' = (r(x) - \frac{P}{K} - \delta(x)Y)\, P\\
		Y' = ( -\beta + \beta \delta(x)P) \,Y
	\end{cases}\;,
\end{equation}
where
\begin{equation}\label{defr}
	r(x) = ax^2 + (b+c)x(1-x) + d(1-x)^2   
\end{equation}
and 
\begin{equation}\label{defdelta}
	\delta(x) = a'x^2 + (b'+c')x(1-x) + d'(1-x)^2\;.
\end{equation}

Eqs. \eqref{xpyBW} are clearly equivalent to the BW model \eqref{BWmodel}. The latter two equations in \eqref{xpyBW} are a frequency-dependent version of the LVLPP model \eqref{modlvlpp}, see also Appendix \ref{app:lvlpp}, with $r(x)$ being interpreted as a frequency-dependent prey birth rate and 
$\delta(x)$ as a frequency-dependent capture rate.

Notice that the frequency-dependent prey birth rate is just a population average of the fitnesses for A and B prey when predators are absent, i.e.
\begin{equation}  \label{xRaver}
	r(x)= x f_A(x,0)+(1-x) f_B(x,0)\;.
\end{equation}
In a similar way, $\delta(x)$ is the population average of the capture rates for A and B prey.

As a consequence, both $r(x)$ and $\delta(x)$ are \emph{strictly positive} for all $x \in[0,1]$.

It is also worth noting that $x(t)=0$ and $x(t)=1$ are solutions to the first of \eqref{xpyBW}. It follows, by uniqueness of solutions, that for any solution of \eqref{xpyBW}, $x(t)$ will be in $(0,1)$ if the initial condition $x(0)$ is in $(0,1)$.

In Mathematical Biology the term \emph{Lotka-Volterra models} refers not only to the predator-prey models we have already seen here, but, more generally, to systems of ODEs of the form
\[x'_i = f_i(x_1, \dots,  x_n) x_i\;,\]
$i=1, 2, \dots, n$, i.e. Kolmogorov systems \cite{frewal}, where the $f_i$ functions are \emph{polynomials of degree at most 1}. Thus, technically speaking, although the LVLPP in \eqref{modlvlpp} is a Lotka-Volterra model, its frequency-dependent version \emph{is not} such. However, as we will see, knowledge of the strict LVLPP model will be useful for proving results on the BW model. This is why Appendix \ref{app:lvlpp} is devoted to the LVLPP model.

\subsection{The predators' famine theorem} \label{sub:fam}
As in the LVLPP model, see Theorem \ref{thlvlpp}, if $K$ is too small, predators will starve due to lack of prey. Although we may have more precise bounds for some of the dynamic scenarios, a general result for predators' famine can be proved:
\begin{theorem}[Predators' famine] \label{thfamine}
	If 
	\begin{equation} \label{faminecond}
		K \, \max_{x \in [0,1]} r(x) < \min_{x\in [0,1]}\frac{1}{\delta(x)}\;,
	\end{equation} predators will go extinct when $t \rightarrow \infty$.
\end{theorem}
\begin{proof}
	At any instant $t$ we may talk of the ``instant nullclines", which are the straight lines $N_1(t): r(x(t)) - \frac{P}{K} - \delta(x(t)) Y=0$ and $N_2(t): P=\frac{1}{\delta(x(t))}$ in the $PY$-plane. Under condition \eqref{faminecond}, we see that the nullclines never intersect in the first quadrant, for any value of $t$. We are thus, for any instant $t$ in which the solution is well-defined, in the situation illustrated by the left panel in Fig. \ref{fig:lvlpporb}. The proof is then a modification of the part of the proof of Theorem \ref{thlvlpp} referring to this left panel.
	
	Let $R=\max_{x \in [0,1]}r(x)$ and $D=\max_{x\in[0,1]} \delta(x)$. Condition \eqref{faminecond} implies that the nullcline $N_1(t)$ always stands to the left of the line $P=\frac{1}{D}$ in the first quadrant. 
	
	Let $x_0 \in [0,1]$, $Y_0>0$ and $P_0>\frac{1}{D}$ and consider the solution of \eqref{xpyBW} with initial condition $(x_0,P_0,Y_0)$. While the projection onto the $PY$-plane of the orbit of this solution lies in the region to the right of $P=\frac{1}{D}$, we have
	\[Y'(t) \leq \beta(-1+ D P_0)Y\]
	and
	\[P'(t)\leq (R- \frac{1}{D K})P\;.\]
	As in the corresponding part of the proof of Theorem \ref{thlvlpp}, we may use the Gronwall inequality and show that $Y(t)$ does not tend to $\infty$ in finite time and that $P(t)$ decays exponentially, forcing the orbit to enter the region to the left of $P=\frac{1}{D}$. In this region $Y$ is decreasing, bounded below by 0 and $P$ cannot be larger than $1/D$. Thus, the solution must be defined up to infinite time and $Y(t)$ must tend to 0 as $t \rightarrow \infty$.
\end{proof}

\subsection{Equilibria with one or two species} \label{sub:equil12}
From now on, we will present many results on equilibrium points of the BW model and their stability. We start by stating a general notation for these equilibria: we will denote an equilibrium solution by the names of the species types present in that solution. For example, we will talk of the A equilibrium, AB equilibrium, ABY equilibria, and so on.

As the only exception, notice that if $A=B=Y=0$, the right-hand sides of all equations in \eqref{BWmodel} vanish. This means that $(A,B,Y)=(0,0,0)$ is an equilibrium solution of the BW model, but it is an equilibrium with no species present. The reader may easily check that it is always unstable, i.e. solutions of the BW model generically do not lead to extinction of the three species. For this reason, the $(0,0,0)$ equilibrium will not be important in our further analysis. 

Before talking about the other equilibria, another convention will be useful:
\begin{definition}
	An equilibrium solution for the BW model will be called \emph{positive} if the population of all species that should be present in that type of equilibrium is strictly positive and the remaining species have non-negative populations.
\end{definition}

Observe that the above definition allows for zero population for some species in some positive equilibria. For example, a positive AY equilibrium solution must have positive coordinates for A and Y species, but have 0 coordinate for the B species. 

The A equilibrium can be found by putting $B=Y=0$ in \eqref{BWmodel}. The right-hand side of the first equation will be 0 if $A=Ka$. The point
\[(A,B,Y)=(Ka,0,0)\]
is the \emph{A equilibrium} of the BW model. It may also be written in $(x,P,Y)$ coordinates as 
\[(x,P,Y)=(1,Ka,0)\;.\]
The reader may notice that this is the only equilibrium with only A individuals and that it is always positive.

Analogously, the model always admits a single positive \emph{B equilibrium} given by
\[(A,B,Y)=(0,Kd,0) \;\;\; \mathrm{or}\;\;\; (x,P,Y)=(0,Kd,0)\;.\]

The A and B equilibria may be stable or unstable. Indeed, when there are no predators, the BW model reduces to the standard replicator equation, and the A and B equilibria may be stable or unstable in this dynamics, depending on the scenario in Table \ref{tab:repsce}.

Still with $Y=0$, the BW model may admit an \emph{AB equilibrium}. It can be found more easily in $(x,P,Y)$ coordinates. Putting $Y=0$, $x \neq 0$, $x \neq 1$ in \eqref{xpyBW}, we have $x'=0$ if 
\[f_A(x,0)=f_B(x,0)\;.\]
Solving this equation, we get $x=x_R$, where
\begin{equation}  \label{defxR}
	x_R= \frac{d-b}{a-c+d-b}\;.
\end{equation}
It can be seen that $x_R \in(0,1)$, i.e. both A and B individuals are present, if and only if 
\[a>c  \;\;\; \mathrm{and} \;\;\; d>b\;,\]
the codominance scenario in Table \ref{tab:repsce}, or
\[a<c  \;\;\; \mathrm{and} \;\;\; d<b\;,\]
the coexistence scenario in Table \ref{tab:repsce}.

By the second equation in \eqref{xpyBW}, the value of $P$ in the AB equilibrium is $Kr(x_R)$. By \eqref{xRaver},
\[r(x_R)= x_R f_A(x_R,0)+(1-x_R) f_B(x_R,0)\;,\]
which is more simply given by either $f_A(x_R,0)$ or $f_B(x_R,0)$. Thus, the coordinates of the AB equilibrium are
\[(x,P,Y)=(x_R,K f_A(x_R,0),0)\] or
\[(A,B,Y)=(K x_R f_A(x_R,0),K (1-x_R) f_A(x_R,0),0)\;.\]
Again, because there are no predators in the AB equilibrium, we know, from results on the standard replicator equation, that the AB equilibrium may be unstable, and that, restricted to the AB plane, it will be stable only in the coexistence scenario of Table \ref{tab:repsce}. 

We already know that when B prey are absent, the BW model reduces to the LVLPP model. Therefore, we have an \emph{AY equilibrium} given by
\[(A,B,Y)= (\frac{1}{a'}, 0, \frac{1}{a'}(a- \frac{1}{Ka'}))\;\;\;\mathrm{or}\;\;\; (x,P,Y)= (1,\frac{1}{a'},  \frac{1}{a'}(a- \frac{1}{Ka'}))\;. \]
Observe that the AY equilibrium is positive only if
\begin{equation} \label{posAY}
	K > \frac{1}{a a'}\;,
\end{equation}
in accordance with the fact, see Appendix \ref{app:lvlpp}, that the LVLPP model does not have a positive coexistence equilibrium if $K$ is too small.

Analogously, we also have a BY equilibrium given by
\[(A,B,Y)= (0,\frac{1}{d'}, \frac{1}{d'}(d- \frac{1}{Kd'}))\;\;\;\mathrm{or}\;\;\; (x,P,Y)= (0,\frac{1}{d'},  \frac{1}{d'}(d- \frac{1}{Kd'}))\;, \]
with
\begin{equation} \label{posBY}
	K > \frac{1}{d d'}\;
\end{equation}
as positivity condition.

From our knowledge of the LVLPP model, whenever the AY equilibrium is positive, it is stable \emph{when restricted to the $AY$ plane}. The analog holds for the BY equilibrium.

To derive conditions for stability or instability of AY, BY and AB equilibria, we recall that the term in parentheses in each of the equations \eqref{BWmodel} may be interpreted as the \emph{per capita} growth rate of the corresponding species, taking into account births and deaths, but also competition among prey and the (positive or negative) effect of predation. 

In order to discern stability or instability of the AY equilibrium, not only restricted to the $AY$ plane, we must look at the sign of the \emph{per capita} growth rate of species B \emph{calculated at the AY equilibrium}. If positive, then the AY equilibrium is unstable; if negative, the AY equilibrium is stable. This result and its analog for the BY equilibrium were already established by Braga and Wardil \cite{bw}, and we derive them here for completeness.

The per capita growth rate of B, see \eqref{BWmodel}, is
\[\frac{cA + dB}{A + B} - \frac{A + B}{K} - \frac{c'A + d'B}{A + B}Y\;.\]
Substituting into the above expression the coordinates of the AY equilibrium, its stability condition is
\[c- \frac{1}{Ka'}- \frac{c'}{a'}(a- \frac{1}{Ka'})<0\;.\]
Multiplying by $a'$, we recover the result presented in \cite{bw} that the AY equilibrium is asymptotically stable whenever it is positive and also
\begin{equation} \label{AYstabil}
	a'c-ac'-\frac{1}{K}(1-\frac{c'}{a'})<0\;.
\end{equation}
If the inequality in \eqref{AYstabil} is reversed, then a positive AY equilibrium is unstable.

Also appearing in BW, we have analogous conditions, with analogous derivation, for a positive BY equilibrium: it is asymptotically stable if positive and
\begin{equation} \label{BYstabil}
	d'b-db'-\frac{1}{K}(1-\frac{b'}{d'})<0
\end{equation}
and unstable if the inequality is reversed.

Although conditions for the stability or instability of the AB  equilibrium were not given in \cite{bw}, we may use the same ideas. A small difference is that the AB equilibrium is not necessarily stable when positive and restricted to the AB plane. In order for the AB equilibrium to be positive, we must have either $a>c$ and $d>b$ (codominance, see Table \ref{tab:repsce}), or $a<c$ and $d<b$ (coexistence). In the first case, the AB equilibrium is unstable. In the second case, it is asymptotically stable restricted to the AB plane.

Thus, the AB equilibrium is asymptotically stable if $a<c$, $d<b$ and also the per capita growth rate of Y is negative when calculated at the AB equilibrium. This leads to the additional condition for stability of the AB equilibrium
\begin{equation} \label{ABstabil}
	K r(x_R) \delta(x_R)-1<0\;.
\end{equation}
If, on the contrary, either $a>c$ and $d>b$, or  $a<c$, $d<b$ and the inequality is reversed in \eqref{ABstabil}, then the AB equilibrium is unstable.

Although the conditions above seem ready to use, we will still rewrite them in a more useful form in Theorems \ref{theogtAY}, \ref{theogtBY} and \ref{theoABstab} ahead. Before that, we must still learn more about ABY equilibria.

\subsection{ABY equilibria}
It is difficult to find ABY equilibria by using the original form of the BW model's equations, \eqref{BWmodel}. Instead, it is a simple task in the reformulation \eqref{xpyBW}.

In fact, equating to 0 the right-hand side of the first of \eqref{xpyBW}, disregarding the solutions $x=0$ and $x=1$, which do not lead to equilibria with all three species, we see that an ABY equilibrium must satisfy
\[Y=m(x)\;,\]
where
\begin{equation} \label{defm}
	m(x) = \frac{(a-c)x+(b-d)(1-x)}{(a'-c')x+(b'-d')(1-x)}\;.
\end{equation}

Equating to 0 the right-hand sides of the last two equations in \eqref{xpyBW} and replacing $Y$ by $m(x)$, we get $P = K(r(x) -\delta(x)m(x))$ and $P = \frac{1}{\delta(x)}$. Equating these two expressions for $P$, we see that the $x$ coordinate of an ABY equilibrium must satisfy
\begin{equation} \label{ABYeq}
	K F(x)= \frac{1}{\delta(x)}\;,
\end{equation}
where
\begin{equation} \label{defF}
	F(x)= r(x) -\delta(x)m(x)\;.  
\end{equation}

The following theorem summarizes these results:
\begin{theorem}  \label{theoABYeq}
	Let $x^* \in (0,1)$ be such that
	\[m(x^*)>0\;\;\;\mathrm{and}\;\;\;K F(x^*)= \frac{1}{\delta(x^*)}\]
	for some positive $K$.
	
	Then $(x,P,Y)=(x^*,\frac{1}{\delta(x^*)},m(x^*))$ is a positive ABY equilibrium.
\end{theorem}

It is necessary, for a positive ABY equilibrium, that $x \in (0,1)$ and that $m(x)$ is positive. In Table \ref{tab:posm}, a consequence of elementary inequality solving, we display -- in each scenario -- the subsets of $(0,1)$ such that $m(x)$ is positive. In order to understand the table, in analogy with $x_R$ defined in \eqref{defxR}, we define $x_P$ as the point in which the denominator in the right-hand side of \eqref{defm} vanishes, i.e. 
\begin{equation}  \label{defxP}
	x_P= \frac{d'-b'}{a'-c'+d'-b'}\;.
\end{equation}
Of course, $x_P \in (0,1)$ if and only if $a'>c'$ and $d'>b'$ (coexistence in predation) or $a'<c'$ and $d'<b'$ (codominance in predation). For convenience, we also define
\begin{equation}  \label{defx1x2}
	x_1 = \min\{x_R,x_P\}  \;\;\;\mathrm{and} \;\;\; x_2 = \max\{x_R,x_P\}\;.
\end{equation}

\begin{table}     
	\begin{tabular}{ccccc}
		\hline\noalign{\smallskip}
		&$\begin{array}{c}
			\textrm{A dominance}\\
			a'<c'\\
			b'<d'
		\end{array}$&$\begin{array}{c}
			\textrm{Codominance}\\
			a'<c'\\
			b'>d'
		\end{array}$ &$\begin{array}{c}
			\textrm{Coexistence}\\
			a'>c'\\
			b'<d'
		\end{array}$ & $\begin{array}{c}
			\textrm{B dominance}\\
			a'>c'\\
			b'>d'
		\end{array}$ \\
		
		\noalign{\smallskip}\hline\noalign{\smallskip}
		$\begin{array}{c}
			\textrm{A dominance}\\
			a>c\\
			b>d
		\end{array}$ &$\varnothing$&$(0,x_P)$&$(x_P,1)$&$(0,1)$ \\
		\noalign{\smallskip}\hline
		$\begin{array}{c}
			\textrm{Codominance}\\
			a>c\\
			b<d
		\end{array}$ &$(0,x_R)$&$(x_1,x_2)$&$(0,x_1)\cup (x_2,1)$&$(x_R,1)$ \\
		\noalign{\smallskip}\hline
		$\begin{array}{c}
			\textrm{Coexistence}\\
			a<c\\
			b>d
		\end{array}$ &$(x_R,1)$&$(0,x_1)\cup (x_2,1)$&$(x_1,x_2)$&$(0,x_R)$ \\
		\noalign{\smallskip}\hline
		$\begin{array}{c}
			\textrm{B dominance}\\
			a<c\\
			b<d
		\end{array}$ &$(0,1)$&$(x_P,1)$&$(0,x_P)$&$\varnothing$ \\
		\noalign{\smallskip}\hline \smallskip
	\end{tabular}
	\caption{\label{tab:posm} The subsets of $(0,1)$, for each of the 16 possible evolutionary scenarios, in which $m(x)>0$. $x_R$ is defined in \eqref{defxR}, $x_P$ in \eqref{defxP}, and $x_1,x_2$ in \eqref{defx1x2}.}
\end{table}

An interesting remark here is that in the two scenarios (1 and 16) in which the same prey type dominates both in reproduction and in predation, positive ABY equilibria cannot exist, because $m(x)<0$ for all $x\in [0,1]$. We can say more on the dynamics in these two scenarios by noticing that the inequalities that characterize them show that in scenario 1 we have $f_A(x,Y)>f_B(x,Y)$ for all $Y\geq 0$, $x \in [0,1]$, and in scenario 16 the reverse inequality holds. In the first case $x$ is always increasing, see \eqref{repl}, whereas it is decreasing in the second case. If the solutions of the ODEs do not tend to infinity in finite time, we have, respectively, extinction of B prey, or extinction of A prey, and the BW model becomes asymptotically the LVLPP model. For large values of $K$, see Theorem \ref{thlvlpp}, we also have survival of the predators. But in both scenarios 1 and 16 it is impossible to have long-term coexistence of the three species, either in equilibrium or otherwise.

Theorem \ref{theoABYeq} teaches us how to calculate ABY equilibria, but we still lack general results, such as those at the end of Subsection \ref{sub:equil12}, about stability or instability of the ABY equilibria. Despite that, in numerical calculations, we may determine the stability of ABY equilibria by evaluating the eigenvalues of the Jacobian matrix at the equilibrium. As will be seen in examples in the next section, for given values of the parameters, we may have no ABY equilibrium, or a single, or multiple ABY equilibria, stable and unstable.

\section{Examples of ABY equilibria}  \label{sec:ABYeq}
As already mentioned, for given values for the parameters, calculating ABY equilibria is in principle straightforward: solve equation \eqref{ABYeq}. By eliminating all denominators, solving \eqref{ABYeq} amounts to finding zeros of a 5th degree polynomial. We must therefore use numerical methods to approximate solutions, or graphical methods to visualize the number of real solutions in $(0,1)$. As examples will show, we may have many different shapes for the graphs in $[0,1]$ of $F(x)$ and $1/\delta(x)$, so even the number of biologically meaningful solutions will depend on the chosen values for the parameters.

Reminding the reader that $\delta(x)$ is strictly positive in $[0,1]$, we may view \eqref{ABYeq} as telling us that for each $x \in (0,1)$ such that $m(x)>0$ and $F(x)>0$, we determine a unique value
\begin{equation}  \label{defkappa}
	\kappa(x)= \frac{1}{\delta(x)F(x)}
\end{equation}
such that, if $K=\kappa(x)$, we have a positive ABY equilibrium with coordinate $x$. Fixing values for all parameters, except $K$, we can plot a graph of $\kappa(x)$. The number of positive ABY equilibria for a given value of $K$ is the number of intersections of the horizontal line with ordinate $K$ with the graph of $\kappa(x)$ for $x\in(0,1)$, $x$ such that $m(x)>0$, see Table \ref{tab:posm}. We will present a few examples that we found interesting and representative, and that will be further explored ahead. The reader is invited to try others for himself/herself. 

In particular, many of the following examples belong to scenario 4 in Table \ref{tab:reppresce}. As we mentioned in the very beginning of this paper, a curiosity about peacocks was among the motivations of Braga and Wardil \cite{bw} when proposing their model. If prey species A stands for real peacocks, with attractive tails which make them dominant in reproduction, but dominated in predation, then scenario 4 might be termed the \emph{peacock scenario}. Prey species B is then an imaginary type of peacock with less attractive tails. Of course, scenario 13 is analogous to 4, with B prey standing for real peacocks. Scenario 4 is not only close to the original motivation -- the peacock question posed by Braga and Wardil -- but is also very rich from a mathematical standpoint. Another advantage of scenario 4, see Table \ref{tab:posm}, is that we do not need to bother about the positivity of $m$, which holds in the whole interval.

In all examples in this paper we will use $\beta=2$. In fact, varying $\beta$ will increase or decrease time scales, but will not alter any of the equilibria.

Just to briefly mention one possibility, in our first example we choose
\begin{equation}
	M_R= \begin{pmatrix}
		1&1.2\\
		0.8&0.6
	\end{pmatrix} \;\;\;\mathrm{and}\;\;\;M_P= \begin{pmatrix}
		0.45&0.55\\
		0.4&0.45
	\end{pmatrix} \;.
\end{equation}
We are in the peacock scenario 4. By plotting a graph, the reader may see that with the above choices we have $F(x)<0$ for all $x \in [0,1]$. As a consequence, there are no positive ABY equilibria for any positive value of $K$. This means that with the above parameter choice there is no possibility of coexistence in equilibrium, either stable or unstable, of the two peacock types.

Let us now proceed to more interesting examples.
\subsection{Scenario 4, $F$ changes sign} \label{sub:ex2}
We choose
\begin{equation}
	M_R= \begin{pmatrix}
		3.9&0.9\\
		0.8&0.7
	\end{pmatrix} \;\;\;\mathrm{and}\;\;\;M_P= \begin{pmatrix}
		1.5&1.5\\
		0.75&0.4
	\end{pmatrix} \;.
\end{equation}

With these values, $F$ has a single root $x^* \approx 0.39$ in $[0,1]$, being positive for $0<x<x^*$ and negative for $x^*<x<1$. As a consequence, $\kappa(x) \rightarrow \infty$ when $x \nearrow x^*$. As seen in its graph, Fig. \ref{fig:kappaex}(a), $\kappa$ has also a local minimum in $\tilde{x}\approx 0.10$. With this information, we have
\begin{itemize}
	\item no positive ABY equilibrium for $K \in (0,\kappa(\tilde{x}))$;
	\item two positive ABY equilibria for $K \in (\kappa(\tilde{x}), \kappa(0))$;
	\item one positive ABY equilibrium for $K  \in (\kappa(0), \infty)$.
\end{itemize}
\begin{figure}
	\centering
	\includegraphics[width=0.9\linewidth]{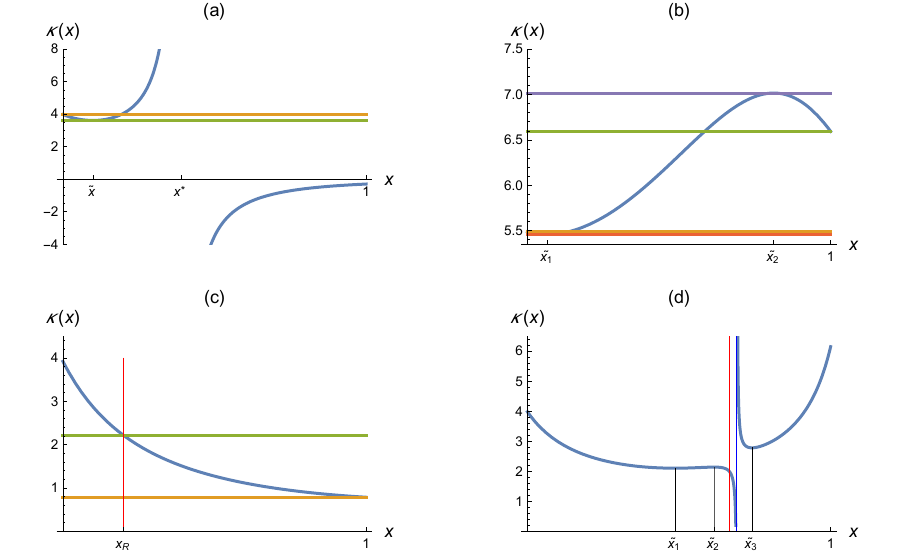}
	\caption{Graphs of the function $\kappa(x)$ in the examples of subsections \ref{sub:ex2} to \ref{sub:ex5}. We use $\beta=2$ and the values of the remaining parameters and the meaning of all symbols are given in the text of each subsection.}
	\label{fig:kappaex}
\end{figure}

We numerically determine the stability or instability of the ABY equilibria by evaluating the eigenvalues of the Jacobian matrix of the system at each equilibrium. The numerical results indicate that for $K \in (\kappa(\tilde{x}), \kappa(0))$  the positive ABY equilibrium with coordinate $x<\tilde{x}$ is asymptotically stable, and the other, with $x>\tilde{x}$, is unstable. The only positive ABY equilibrium for $K  \in (\kappa(0), \infty)$ is numerically unstable.

Biologically, in this example we have a small range in $K$ such that a stable coexistence equilibrium for the two peacock types and predators exists.

\subsection{Scenario 4, $F$ is positive in $[0,1]$.}\label{sub:ex3}
The next parameter choice, still in scenario 4, is
\begin{equation}
	M_R= \begin{pmatrix}
		3.77&0.9\\
		1.2&0.54
	\end{pmatrix} \;\;\;\mathrm{and}\;\;\;M_P= \begin{pmatrix}
		1.565&1.5\\
		0.47&0.535
	\end{pmatrix} \;.
\end{equation}

The reader may see that $F$ is positive in all $[0,1]$. This means that $\kappa$ is also positive and we have a positive ABY equilibrium at each $x \in (0,1)$. The graph of $\kappa(x)$ is shown in Fig. \ref{fig:kappaex}(b). The function has a minimum at $\tilde{x}_1 \approx 0.07$ and a maximum at $\tilde{x}_2 \approx 0.81$. The results are the following, in which stability or instability was decided by numerical evaluation of the eigenvalues:
\begin{itemize}
	\item no positive ABY equilibrium for $K \in (0,\kappa(\tilde{x}_1))$;
	\item two positive ABY equilibria for $K \in (\kappa(\tilde{x}_1), \kappa(0))$; the one in $(0,\tilde{x}_1)$ is asymptotically stable, and the other, in $(\tilde{x}_1, \tilde{x}_2)$, is unstable;
	\item one positive ABY equilibrium for $K  \in (\kappa(0), \kappa(1))$; it is unstable;
	\item two positive ABY equilibria for $K \in (\kappa(1), \kappa(\tilde{x}_2))$; the one in $(\tilde{x}_1, \tilde{x}_2)$ is unstable, and the other, in $(\tilde{x}_2,1)$, is asymptotically stable when $K$ is close to $\kappa(1)$, but becomes unstable at $K \approx 6.8$;
	\item no positive ABY equilibrium for $K \in (\kappa(\tilde{x}_2), \infty)$.
\end{itemize}

\subsection{Scenario 8} \label{sub:ex4}
In scenario 8 we have $x_R \in(0,1)$, and, according to Table \ref{tab:posm}, $m(x)>0$ for $x \in (x_R,1)$.

We choose
\begin{equation}
	M_R= \begin{pmatrix}
		3.2&0.75\\
		3&0.8
	\end{pmatrix} \;\;\;\mathrm{and}\;\;\;M_P= \begin{pmatrix}
		0.5&0.6\\
		0.35&0.3
	\end{pmatrix} \;.
\end{equation}
A graph of $\kappa(x)$ is shown in Fig. \ref{fig:kappaex}(c). By the graph, it can be seen that we have a single positive ABY equilibrium for $K \in (\kappa(1), \kappa(x_R))$ and no positive ABY equilibrium either for $K \in(0,\kappa(1))$ or for $K> \kappa(x_R)$. The positive ABY equilibrium is numerically seen to be asymptotically stable in a large part of its existence range. But it becomes unstable when $K \approx 1.66$, close to $\kappa(x_R)$.  

\subsection{Scenario 10}\label{sub:ex5}
In this scenario, we have both $x_R$ and $x_P$ in $(0,1)$. We choose 
\begin{equation}
	M_R= \begin{pmatrix}
		0.7&0.9\\
		0.8&0.7
	\end{pmatrix} \;\;\;\mathrm{and}\;\;\;M_P= \begin{pmatrix}
		0.25&1.5\\
		0.75&0.4
	\end{pmatrix} \;.
\end{equation}
In Fig. \ref{fig:kappaex}(d) we show the graph of $\kappa(x)$ along with the vertical lines $x=x_R$ (in red) and $x=x_P$ (in blue). According to Table \ref{tab:posm}, the set of values for $x$ such that $m(x)>0$ excludes the region between $x_R$ and $x_P$. In the allowed region $\kappa$ has three local extrema $\tilde{x}_i$, $i=1, 2, 3$, also indicated in the figure. This creates a rich situation with several intervals for $K$ in which the number of positive ABY equilibria varies. We refrain from indicating all intervals, but in some of them we have as many as three positive ABY equilibria for the same value of $K$, two being asymptotically stable and one unstable.

\subsection{Some conclusions}
All four panels in Fig. \ref{fig:kappaex} tell different versions of some ``typical stories" of creation and annihilation of positive ABY equilibria. We may think of $K$, the carrying capacity control parameter, as starting with value 0, and being gradually increased up to $\infty$. As $K$ passes through some threshold values, creations or annihilations of positive ABY equilibria occur; sometimes a pair of positive ABY equilibria is simultaneously created or annihilated. Increasing the value of $K$ has an important biological interpretation: it means increasing carrying capacities of both prey types, hence providing more food to predators, thus enhancing the importance of the predation matrix parameters with respect to the reproduction matrix parameters.

As an example, consider the graph in Fig. \ref{fig:kappaex}(b). When $K$ starts from 0, there is no positive ABY equilibrium. When $K$ passes through the local minimum value $\kappa(\tilde{x}_1)$, a pair of positive ABY equilibria is created. By increasing $K$, the $x$ coordinate of one of these equilibria decreases, and the $x$ coordinate of the other increases. Let $\overline{x}_1$ be the location of the $x$-decreasing equilibrium and $\overline{x}_2$ be the location of the $x$-increasing equilibrium. Further increasing $K$, when $K=\kappa(0)$, $\overline{x}_1$ reaches 0 and, due to the fact that $\kappa'(0)<0$ in that example, the corresponding ABY equilibrium becomes non-positive for $K>\kappa(0)$; in other words, a positive ABY equilibrium is annihilated at $x=0$. Meanwhile $\overline{x}_2$ continues to increase. When $K=\kappa(1)$, a new positive ABY equilibrium appears at $x=1$ and its coordinate $\overline{x}_3$ decreases. More precisely, a non-positive ABY equilibrium became positive when $K=\kappa(1)$. When $K$ passes through the local maximum $\kappa(\tilde{x}_2)$ the pair of equilibria with coordinates $\overline{x}_2$ and $\overline{x}_3$ are simultaneously annihilated. For larger values of $K$, there will be no other positive ABY equilibria.

The other panels in Fig. \ref{fig:kappaex} tell similar stories. Panels (c) and (d) in Fig. \ref{fig:kappaex} also show that the creation or annihilation of a single positive ABY equilibrium may occur not only in $x=0$ or $x=1$, but also in $x=x_R$.

\section{Stability of the AY, BY and AB equilibria} \label{sec:stabAYBYAB}
We had promised to return to the conditions of stability and instability of the AY, BY and AB equilibria after learning more about the ABY equilibria. The important fact is that, surprisingly,  conditions such as \eqref{AYstabil}, \eqref{BYstabil} and \eqref{ABstabil} provide a link between the change in the stability status of these equilibria and the creation/annihilation of a positive ABY equilibrium.

In fact, suppose at first that a positive ABY equilibrium is annihilated or created at $x=1$. Strictly speaking, coordinate $x$ of a positive, biologically meaningful, ABY equilibrium must be a solution in $(0,1)$ of \eqref{ABYeq}. But the creation or annihilation at $x=1$ of an ABY equilibrium means that $x=1$ solves \eqref{ABYeq}, i.e. $KF(1)-1/\delta(1)=0$.

Using the definitions \eqref{defF} and \eqref{defdelta}, we see that
\begin{align}
	(a'-c')(KF(1)-\frac{1}{\delta(1)})&= (a'-c') \left[K(a-\frac{a-c}{a'-c'}\, a')-\frac{1}{a'} \right]
	\nonumber\\ 
	&= K \left[ a'c-ac'-\frac{1}{K}(1-\frac{c'}{a'})\right]\;.\end{align}
This means that the stability condition \eqref{AYstabil} for the AY equilibrium can be rewritten as
\begin{equation} \label{altAYstab}
	(a'-c') (KF(1)-\frac{1}{\delta(1)}) <0
\end{equation}
and similarly for the instability condition. We see that a change in the stability status of the AY equilibrium, either from unstable to stable or the reverse, requires that $KF(1)-1/\delta(1)=0$, i.e., a positive ABY equilibrium must be created or annihilated at $x=1$. 

Similarly, by \eqref{BYstabil}, the BY equilibrium will change its stability status if an ABY equilibrium is created or annihilated at $x=0$. And, by \eqref{ABstabil}, the AB equilibrium will change from stable to unstable (and never the reverse if $K$ is being increased) only when an ABY equilibrium is created or annihilated at $x=x_R$.

As we shall see, rewriting the stability and instability conditions for the AY equilibrium will shed strong light on the game-theoretic interpretation of its stability or instability. Before stating the result, define
\begin{equation}  \label{defkac}
	K^A_C= \kappa(1)= \frac{1}{a'F(1)}\;.
\end{equation}
Notice also that we may rewrite
\begin{align}
	F(1)= & \,a-\frac{a-c}{a'-c'}\, a' \label{F1<a}\\
	=& \,\frac{c c'}{a'-c'}\, \left(\frac{a'}{c'}- \frac{a}{c}\right)\label{F1>0}\;.
\end{align}
Finally, recall the positivity condition \eqref{posAY} for the AY equilibrium.

Then,
\begin{theorem} \label{theogtAY}
	Suppose all parameters in the BW model are positive and $a \neq c$, $a' \neq c'$. 
	\begin{enumerate}
		\item If $\frac{a'}{c'}>\frac{a}{c}>1$, then $K^A_C>\frac{1}{a a'}$, and the AY equilibrium is asymptotically stable if $K \in (\frac{1}{a a'}, K^A_C)$ and unstable if $K \in (K^A_C, \infty)$.
		\item If either $\frac{a}{c}>\frac{a'}{c'}>1$ or $\frac{a}{c}>1>\frac{a'}{c'}$, then $K^A_C<\frac{1}{a a'}$, and the AY equilibrium is asymptotically stable whenever positive, i.e., $K \in (\frac{1}{a a'}, \infty)$.
		\item If $\frac{a'}{c'}<\frac{a}{c}<1$, then $K^A_C>\frac{1}{a a'}$, and the AY equilibrium is unstable if $K \in (\frac{1}{a a'}, K^A_C)$ and asymptotically stable if $K \in (K^A_C, \infty)$.
		\item If either $\frac{a}{c}<\frac{a'}{c'}<1$ or $\frac{a}{c}<1<\frac{a'}{c'}$, then $K^A_C<\frac{1}{a a'}$, and the AY equilibrium is unstable whenever positive, i.e., $K \in (\frac{1}{a a'}, \infty)$.
	\end{enumerate}
\end{theorem}
\begin{proof}
	The stability of the AY equilibrium is determined by inequality \eqref{altAYstab}. We may solve it for $K$. The solution is $K<K^A_C$ if $a'>c'$ and $K>K^A_C$ if $a'<c'$.
	
	Suppose now that $\frac{a'}{c'}>\frac{a}{c}>1$. Then \eqref{F1>0} shows that $F(1)>0$, and \eqref{F1<a} shows that $F(1)<a$. These facts prove that $K^A_C>\frac{1}{a a'}$. Putting together positivity and the solution of \eqref{altAYstab}, we see that the AY equilibrium is asymptotically stable if $K \in (\frac{1}{a a'}, K^A_C)$. If $K>K^A_C$, then the positivity condition is automatically satisfied and also the analog of \eqref{altAYstab} for instability.
	
	If $\frac{a}{c}>\frac{a'}{c'}>1$, \eqref{F1>0} shows that $F(1)<0$. Using this in \eqref{altAYstab} shows that it is satisfied for any positive $K$. Thus, the AY equilibrium is asymptotically stable whenever positive.
	
	If $\frac{a}{c}>1>\frac{a'}{c'}$, \eqref{F1<a} shows that $F(1)>a$, which results in $0<K_C^A<\frac{1}{a a'}$. Satisfying the positivity condition will, consequently, automatically satisfy the condition $K>K^A_C$. But now $a'<c'$, and the solution of \eqref{altAYstab} becomes $K>K^A_C$. So, satisfying only $K>\frac{1}{a a'}$ secures the asymptotic stability of the AY equilibrium.
	
	The proof of the remaining two statements is completely analogous.    
\end{proof}

We may now interpret the results of Theorem \ref{theogtAY} in the light of Game Theory. 

The condition $a>c$ means that $f_A(1,0)>f_B(1,0)$, which, by continuity, implies $f_A(x,0)>f_B(x,0)$ for $x$ close enough to 1. We may then say that $a>c$ means that \emph{A prey dominate in reproduction whenever the A fraction is large enough}. On the contrary, $a<c$ means that \emph{B prey dominate in reproduction whenever the A fraction is large enough}. 

Similarly, $a'<c'$ means that when the predator population $Y$ is large enough, then $f_A(1,Y)>f_B(1,Y)$. By continuity, this implies $f_A(x,Y)>f_B(x,Y)$ for $x$ close enough to 1 and large enough $Y$. In game-theoretic terms, $a'<c'$ means that \emph{A prey dominate in predation whenever the A fraction is large enough}. And $a'>c'$ means that \emph{B prey dominate in predation whenever the A fraction is large enough}.

If $\frac{a'}{c'}>\frac{a}{c}>1$, we may say that, when the A fraction is large enough, A prey dominate in reproduction, B prey dominate in predation, and \emph{the predation dominance is stronger than the reproduction dominance}.

The first two assertions in Theorem \ref{theogtAY} show that if A prey are dominant in reproduction when the A fraction is large, as soon as the AY equilibrium becomes positive, i.e., as soon as $K$ overcomes $\frac{1}{aa'}$, the AY equilibrium is stable. At least for $K$ slightly above the positivity threshold, the reproduction dominance of A prey prevails.

If either A prey also dominate in predation when the A fraction is large, $\frac{a}{c}>1>\frac{a'}{c'}$, or the dominance of B prey in predation when the A fraction is large is weak, $\frac{a}{c}>\frac{a'}{c'}>1$, then the AY equilibrium remains stable for all larger values of $K$. In these situations, a small fraction of B prey is not able to invade a population with A prey and predators in equilibrium. This holds for all values of $K$.

On the other hand, if, when the A fraction is large, A dominates in reproduction, B dominates in predation, and the predation dominance is strong, $\frac{a'}{c'}>\frac{a}{c}>1$, then the AY equilibrium is destabilized when $K$ is larger than $K^A_C$. In this situation, for $K>K^A_C$ even a tiny fraction of B prey can invade a population with A prey and predators in equilibrium. The predation dominance prevails only if $K$ is large and predation dominance is stronger than reproduction dominance.

The latter two assertions in Theorem \ref{theogtAY} deal with situations in which the A fraction is large and B prey dominate in reproduction. In these cases for $K$ not much larger than $\frac{1}{a a'}$, the AY equilibrium is unstable. But it will become stable for larger values of $K$ if $\frac{a'}{c'}<\frac{a}{c}<1$, i.e., A prey dominate in predation and the predation dominance is stronger. In the other cases, the AY equilibrium will remain unstable whenever positive.

We may, in analogy to \eqref{defkac}, define
\begin{equation}
	K^B_C= \kappa(0)= \frac{1}{d' F(0)}\;,
\end{equation}
and prove an analogous theorem for the BY equilibrium in the situations in which the B fraction is large enough. For clarity, we state the result here: 
\begin{theorem} \label{theogtBY}
	Suppose that all parameters in the BW model are positive and $d \neq b$, $d' \neq b'$. 
	\begin{enumerate}
		\item If $\frac{d'}{b'}>\frac{d}{b}>1$, then $K^B_C>\frac{1}{dd'}$, and the BY equilibrium is asymptotically stable if $K \in (\frac{1}{d d'}, K^B_C)$ and unstable if $K \in (K^B_C, \infty)$.
		\item If either $\frac{d}{b}>\frac{d'}{b'}>1$ or $\frac{d}{b}>1>\frac{d'}{b'}$, then $K^B_C<\frac{1}{dd'}$, and the BY equilibrium is asymptotically stable whenever positive, i.e., $K \in (\frac{1}{d d'}, \infty)$.
		\item If $\frac{d'}{b'}<\frac{d}{b}<1$, then $K^B_C>\frac{1}{dd'}$, and the BY equilibrium is unstable if $K \in (\frac{1}{d d'}, K^B_C)$ and asymptotically stable if $K \in (K^B_C, \infty)$.
		\item If either $\frac{d}{b}<\frac{d'}{b'}<1$ or $\frac{d}{b}<1<\frac{d'}{b'}$, then $K^B_C<\frac{1}{dd'},$ and the BY equilibrium is unstable whenever positive, i.e., $K \in (\frac{1}{d d'}, \infty)$.
	\end{enumerate}
\end{theorem}

As for the AY and BY equilibria, we can now show that the left hand side of the stability condition \eqref{ABstabil} can be seen as a positive ABY equilibrium being created or annihilated at $x=x_R$. In fact, as $m(x_R)=0$, $F(x_R)=r(x_R)$. It follows that a possible change in the stability status of the AB equilibrium occurs when $K=\kappa(x_R)$, which is always positive if $x_R \in (0,1)$. This result is summarized in the theorem below; the proof is easy and left to the reader.
\begin{theorem}  \label{theoABstab}
	The AB equilibrium is asymptotically stable if $a<c$, $d<b$ and also
	\begin{equation} \label{ABstabilkappa}
		K < \kappa(x_R)\;.
	\end{equation}
	It is unstable if $a<c$, $d<b$ and 
	\begin{equation} \label{ABunstabkappa}
		K >\kappa(x_R)\;,
	\end{equation}
	or if $a>c$ and $d>b$.
\end{theorem}

In contrast to the AY and BY equilibria, which may be either stabilized or destabilized when $K$ is larger than a threshold $K^A_C$ or $K^B_C$, a positive AB equilibrium is either always unstable (if $a>c$, $d>b$), or becomes unstable for $K> \kappa(x_R)$ (if $a<c$, $d<b$).

\section{Bifurcations and stability or instability of the ABY equilibria}  \label{sec:bif}
Let us now turn to the question of bifurcations and their consequences on stability or instability of positive ABY equilibria. We know that, as $K$ is increased, positive ABY equilibria may be created or annihilated at $x=1$, $x=0$, or $x=x_R$, and, for the same values of $K$, AY, BY and AB equilibria change stability status. We also know that a pair of ABY equilibria may be simultaneously created or annihilated at some threshold values for $K$. These facts suggest that $K$ should be thought of as a \emph{bifurcation parameter} and that the threshold values at which ABY equilibria appear or disappear are \emph{bifurcation points} \cite{guckhol, perko}. 

Appendix \ref{app:sntc} brings some technical information on two among the most common types of local bifurcations: saddle-node (or tangential, or fold) and transcritical.

Simultaneous creation or annihilation of ABY equilibria pairs at a bifurcation point suggests a \emph{saddle-node} bifurcation. Theorems \ref{theogtAY}, \ref{theogtBY} and \ref{theoABstab} show, respectively, that AY, BY and AB equilibria may change their stability status only when an ABY equilibrium merges with them. This suggests \emph{transcritical} bifurcations.

Both saddle-node and transcritical bifurcations occur when the Jacobian matrix at an equilibrium has a zero eigenvalue with multiplicity 1. By calculating the partial derivatives of the field at a general point $(x,P,Y)$, substituting the values of $P$ and $Y$ by their values at a positive ABY equilibrium, see Theorem \ref{theoABYeq}, and $K$ by its value $\kappa(x)=1/(\delta(x)F(x))$, we obtain the Jacobian at a positive ABY equilibrium as a function of the $x$ coordinate of the equilibrium:
\begin{equation}  \label{ABYJac}
	J= \begin{pmatrix}
		0&0& -x(1-x) \left[(a'-c')x+(b'-d')(1-x)\right]\\
		\frac{r'(x)-\delta'(x)m(x)}{\delta(x)}&-F(x)&-1\\
		\beta \frac{m(x) \delta'(x)}{\delta(x)}&\beta \delta(x)m(x)&0
	\end{pmatrix}\;.
\end{equation}

We now calculate the determinant of $J$. The result is
\begin{equation}  \label{detJ}
	\det J=\beta \, \delta(x)\, F(x)^2\, x(1-x)\,\left(f_A(x,0)-f_B(x,0)\right)  \, \kappa'(x)\;.
\end{equation}
This shows that $J$ has a 0 eigenvalue if and only if $x=0$, $x=1$, $x=x_R$, or at critical points of $\kappa(x)$. 

If $x \in (0,1)$ and $m(x)>0$, conditions for the location of a positive ABY equilibrium, then the first row of $J$ is nonzero. This holds in particular for positive ABY equilibria at critical points of $\kappa$. On the other hand, the first row of $J$ vanishes for ABY equilibria at $x=0$ or $x=1$, i.e., ABY equilibria created or annihilated at $x=0$ or $x=1$. For ABY equilibria created or annihilated at $x=x_R$, the third row of $J$ vanishes, because $m(x_R)=0$.

We start by considering the case of a positive ABY equilibrium at a critical point of $\kappa$. Let $w$ be, as in Appendix \ref{app:sntc}, the left eigenvector of $J$ with eigenvalue $0$. The distribution of zero and nonzero elements in $J$ thus shows that the second component in $w$ must be nonzero for a positive ABY equilibrium at a critical point of $\kappa$.

The partial derivatives of the vector field with respect to the bifurcation parameter $K$ are
\begin{equation}  \label{Kder}
	(0,\frac{P^2}{K^2},0)\;.
\end{equation}
Hence, hypothesis \eqref{sn1} for a saddle-node bifurcation is always satisfied for a positive ABY equilibrium at a critical point $\tilde{x}$ of $\kappa$. Another hypothesis for a saddle-node is that the Jacobian has a single eigenvalue equal to 0 at the bifurcation point; it is satisfied if we also suppose $\kappa''(\tilde{x}) \neq 0$ and $\tilde{x} \neq x_R$.

Unfortunately, due to the complicated expressions appearing in the calculation, we were not able to prove that the second inequality \eqref{sn2} for a saddle-node bifurcation is always satisfied for a positive ABY equilibrium at a critical point of $\kappa$. For all the examples in Fig. \ref{fig:kappaex} we numerically calculated the left hand side of \eqref{sn2} and indeed found a nonzero value. But we cannot rule out the possibility that, for some fine-tuned values of the parameters, inequality \eqref{sn2} would become an equality. In this sense, we added to the statement of the upcoming theorem that it holds for \emph{generic} values for the parameters. Its proof, not formally presented, is only the verification of the hypotheses for saddle-node bifurcations stated in Appendix \ref{app:sntc} and commented above.

\begin{theorem}  \label{thsn}
	Let $\tilde{x} \in(0,1)$, $\tilde{x}\neq x_R$ be the fraction of A prey in a positive ABY equilibrium. Suppose also that $\kappa'(\tilde{x})=0$, $\kappa''(\tilde{x}) \neq 0$, and that the parameters have generic values such that inequality \eqref{sn2} is satisfied. Then the BW model has a saddle-node bifurcation at $\tilde{z}=(\tilde{x}, \frac{1}{\delta(\tilde{x})}, m(\tilde{x}))$ when $K=\kappa(\tilde{x})$.
\end{theorem}

The importance of the theorem above is the fact that it can be used to infer the stability or instability of the elements of the pair of ABY equilibria involved in the saddle-node bifurcation. Let $s$ denote the number of eigenvalues (counted with multiplicity) of the Jacobian at $\tilde{z}$ with negative real part, i.e. the dimension of the stable manifold of the equilibrium. Then, according to the theorem by Sotomayor \cite{sotomayor} cited at Appendix \ref{app:sntc}, the equilibria created or annihilated at $K=\kappa(\tilde{x})$ are hyperbolic and have stable manifolds with dimensions $s$ and $s+1$. Of particular interest is the case $s=2$, because it guarantees that one of the equilibria involved must be asymptotically stable (stable manifold with dimension 3), whereas the other is unstable (with stable manifold of dimension 2).

For ABY equilibria occurring at critical points of $\kappa$, we have no general result on the value of $s$, but we can numerically calculate the eigenvalues and discover the value of $s$. The most frequent value in our examples is exactly $s=2$. In the four examples illustrated in Fig. \ref{fig:kappaex} we have six critical points of $\kappa$. Five out of them exhibit $s=2$. The only exception is $\tilde{x}_2$ in panel (b), in which $s=0$. This proves that both ABY equilibria annihilated at $K=\kappa(\tilde{x}_2)$ in that example must be unstable. The reader may verify that in a neighborhood of each other critical point in the examples of Fig. \ref{fig:kappaex} we have exactly one stable and one unstable ABY equilibrium. Moreover, dimensions of the stable manifolds of these equilibria must be preserved unless hyperbolicity is lost due to some other bifurcation.

We now turn to the situations in which a positive ABY equilibrium is created or annihilated at $x=1$, $x=0$, or $x=x_R$. To be more concrete, let us consider at first only the case in which the equilibrium is created or annihilated at $x=1$. Theorem \ref{theogtAY} informs us that whenever that happens, i.e. $K=K^A_C$, the AY equilibrium must change its stability status. 

Our knowledge of the LVLPP model guarantees that the AY equilibrium, when restricted to the AY plane, is always asymptotically stable whenever it is positive. This means that the stable manifold of a positive AY equilibrium has dimension $s \geq 2$. If $s=2$, the equilibrium is unstable; if $s=3$, it is asymptotically stable. If the AY equilibrium changes from unstable to stable, or the contrary, then, besides the other eigenvalues with negative real parts, the third eigenvalue must switch from positive to negative, or the contrary. By continuity of the eigenvalues as functions of $K$, this eigenvalue must be 0 when $K=K^A_C$. This proves that we have a single eigenvalue of the Jacobian in the AY equilibrium equal to 0 at $K=K^A_C$.

If $\kappa'(1)>0$, a positive ABY equilibrium is located close to $x=1$, for $K$ slightly smaller than $K^A_C$, and becomes non-positive, i.e., it is annihilated, for $K>K^A_C$. In all examples we studied, the ABY equilibrium had stability status contrary to the nearby AY equilibrium just before merging with it. After merging, $K>K^A_C$, the AY equilibrium had changed stability status. 

If $\kappa'(1)<0$, a positive ABY equilibrium exists close to $x=1$, for $K$ slightly larger than $K^A_C$. At $K=K^A_C$, we know that the AY equilibrium changes stability status. In all examples we studied, including panels (b) and (c) in Fig. \ref{fig:kappaex}, the stability status of the ABY equilibrium just after its creation is again contrary to the stability status of the nearby AY equilibrium.

These facts suggest that a transcritical bifurcation occurs at $x=1$, $K=K^A_C$. A consequence would be the exchange of stability status between the merging equilibria, as observed in the examples. Let us then see whether the other hypotheses \eqref{tc1}, \eqref{tc2} and \eqref{sn2} are verified.

As the first row of the Jacobian \eqref{ABYJac} vanishes at $x=1$, it is easy to see that the left eigenvector with eigenvalue 0 is $w=(1,0,0)$. The partial derivative of the vector field with respect to the bifurcation parameter is still \eqref{Kder} and the dot product of these vectors vanishes. The first condition \eqref{tc1} for a transcritical bifurcation is thus valid. But we have a problem here: an easy calculation shows that the inequality \eqref{tc2} is always false.

We know that the 1-dimensional vector field $f_{\mu}(x)=\mu^3 x-x^2$ also violates only condition \eqref{tc2} for the transcritical bifurcation, and still the two equilibria which merge when $\mu=0$ exchange their stability. Although we have no theorem to justify it, it seems that we have a similar situation here for the BW model when a positive ABY equilibrium is created or annihilated at $x=1$. We cannot say that a transcritical bifurcation occurs, but something similar, with exchange of stability status between the merging equilibria, is likely to be true. 

Before generalizing into a single conjecture the three cases of positive ABY equilibria being created/annihilated at $x=0$, $x=1$, or $x=x_R$, let us summarize in detail, for the case $x=1$, what we have observed in many examples:
\begin{enumerate}
	\item Annihilation of a positive ABY equilibrium and destabilization of the AY equilibrium: If $\frac{a'}{c'}>\frac{a}{c}>1$ and $\kappa'(1)>0$, then an ABY equilibrium is annihilated at $K=K^A_C$, and it is unstable at least for $K$ slightly smaller than $K^A_C$.
	\item Creation of a positive ABY equilibrium and destabilization of the AY equilibrium: If $\frac{a'}{c'}>\frac{a}{c}>1$ and $\kappa'(1)<0$, then an ABY equilibrium is created at $K=K^A_C$, and it is asymptotically stable at least for $K$ slightly larger than $K^A_C$.
	\item Annihilation of a positive ABY equilibrium and stabilization of the AY equilibrium: If $\frac{a'}{c'}<\frac{a}{c}<1$ and $\kappa'(1)>0$, then an ABY equilibrium is annihilated at $K=K^A_C$, and it is asymptotically stable at least for $K$ slightly smaller than $K^A_C$. 
	\item Creation of a positive ABY equilibrium and stabilization of the AY equilibrium: If $\frac{a'}{c'}<\frac{a}{c}<1$ and $\kappa'(1)<0$, then an ABY equilibrium is created at $K=K^A_C$, and it is unstable at least for $K$ slightly larger than $K^A_C$. 
\end{enumerate}

A completely analogous situation occurs when a positive ABY equilibrium is annihilated or created at $x=0$. In all examples we studied, we saw that the stability status of a positive ABY equilibrium -- just before merging with the BY equilibrium (annihilation), or just after merging with it (creation) -- is contrary to the stability status of the BY equilibrium. Only condition \eqref{tc2} is violated, and we cannot say we have a transcritical bifurcation. 

Something similar also holds for the case when a positive ABY equilibrium is created or annihilated at $x=x_R$. The only difference here is that the third row of the Jacobian \eqref{ABYJac} vanishes, as already remarked, and $w=(0,0,1)$. Condition \eqref{tc1} holds, but \eqref{tc2} is false. 

We may summarize our observations on the change of stability status of AY, BY and AB equilibria and stability of nearby positive ABY equilibria in the following
\begin{conjecture} \label{conjtc}
	
	\begin{enumerate}
		Suppose generic positive values for all parameters in the BW model and also that $\kappa$ does not have a critical point at $x=1$, $x=0$, or $x=x_R$.
		\item If a positive AY (respectively BY or AB) equilibrium changes from stable to unstable when $K$ overcomes a threshold $K^A_C$ (resp. $K^B_C$ or $\kappa(x_R)$), then at $x=1$ (resp. $x=0$ or $x=x_R$) either an unstable positive ABY equilibrium is annihilated, or a stable positive ABY equilibrium is created. The stability or instability of the ABY equilibrium holds for $K$ close to the threshold.
		\item If a positive AY (respectively BY) equilibrium changes from unstable to stable when $K$ overcomes a threshold $K^A_C$ (respectively $K^B_C$), then at $x=1$ (resp. $x=0$) either a stable positive ABY equilibrium is annihilated, or an unstable positive ABY equilibrium is created. The stability or instability of the ABY equilibrium holds for $K$ close to the threshold.
		\item If a positive ABY equilibrium is created or annihilated at $x=x_R$ and the AB equilibrium is unstable for $K<\kappa(x_R)$, then the positive ABY equilibrium must be unstable at least for $K$ close to $\kappa(x_R)$. 
	\end{enumerate}
\end{conjecture}

In any case, there is an exchange of stability status between the positive ABY equilibrium and the other equilibrium involved, as if a transcritical bifurcation had occurred.

Taking as true Conjecture \ref{conjtc} above and using Theorem \ref{thsn}, we can predict the stability of all ABY equilibria appearing in the examples of Fig. \ref{fig:kappaex}, at least for values of $K$ close to the threshold at which they were annihilated or created.

In the example of Subsection \ref{sub:ex2}, panel (a) in Fig. \ref{fig:kappaex}, a saddle-node bifurcation occurs at $K=\kappa(\tilde{x})$. As the dimension of the stable manifold of the ABY equilibrium is $s=2$, one of the ABY equilibria for $K$ slightly larger than $\kappa(\tilde{x})$ should be stable and the other, unstable. As in this example, we have $d'/b'<d/b<1$, the BY equilibrium must change from unstable to stable when $K=K^B_C$. According to Conjecture \ref{conjtc}, it is necessary that the ABY equilibrium annihilated at $K=K^B_C$ is stable. It is reasonable then that the ABY equilibrium with $x<\tilde{x}$, created at $K=\kappa(\tilde{x})<K^B_C$, is the stable one, forcing the other, with $x>\tilde{x}$, to be unstable. 

In the next example, Subsection \ref{sub:ex3}, panel (b) in Fig. \ref{fig:kappaex}, we have two critical points of $\kappa$, so two saddle-node bifurcations. The first saddle-node bifurcation, at $K=\kappa(\tilde{x}_1)$, creates a pair of ABY equilibria. As in the previous example, the equilibrium created at $K=\kappa(\tilde{x}_1)$ has $s=2$. Thus, one of the ABY equilibria created must be stable and the other, unstable. Similarly to the previous example, the BY equilibrium must change from unstable to stable. As this requires annihilation of a stable ABY equilibrium, the simpler explanation is that the ABY equilibrium with smaller $x$ is the stable one. As we also have $a'/c'>a/c>1$, the AY equilibrium will change from stable to unstable at $K=K^A_C$. By Conjecture \ref{conjtc}, the ABY equilibrium created at $K=K^A_C$ must be stable for $K$ slightly larger than $K^A_C$. A curious thing is that this ABY equilibrium, although initially stable as predicted by Conjecture \ref{conjtc}, becomes unstable when $K \approx 6.8$. This loss of stability ``out of nothing", before the second saddle-node bifurcation, will be the subject of the next section. For $K=\kappa(\tilde{x}_2)$ the remaining two positive ABY equilibria are annihilated. As already mentioned, the dimension of the stable manifold of the only positive ABY equilibrium when $K=\kappa(\tilde{x}_2)$ is $s=0$.

We now analyze the example of Subsection \ref{sub:ex4}, panel (c) in Fig. \ref{fig:kappaex}. As $a'/c'>a/c>1$, the AY equilibrium is stable as soon as it becomes positive, but becomes unstable at $K=K^A_C$. The ABY equilibrium created at $K=K^A_C$ is then initially stable. The AB equilibrium is unstable, as it must be when $a>c$ and $d>b$, and, as proved in Theorem \ref{theoABstab}, cannot change from unstable to stable. If the ABY equilibrium created at $K=K^A_C$ remained stable until $K=\kappa(x_R)$, it would contradict the third statement of Conjecture \ref{conjtc}. This forces, as in the previous example, the ABY equilibrium to change from stable to unstable ``out of nothing" for some $K<\kappa(x_R)$. More precisely, this happens when $K \approx 1.66$.

The reader is invited to analyze the more complicated example of Subsection \ref{sub:ex5} in which three saddle-node bifurcations occur.

\section{Hopf bifurcations and limit cycles} \label{sec:hopf}
In the numerical explorations we did with the BW model, sometimes, for some ranges of the parameter $K$, we saw solutions that, after a transient, seemed to converge to periodic functions. In Fig. \ref{fig:periodic}, we illustrate one example of such solutions. The example uses $K=1.7$ and the parameter values described in Subsection \ref{sub:ex4}, illustrated in panel (c) of Fig. \ref{fig:kappaex}. From a biological point of view, such solutions mean that it is possible for the three species to coexist stably not only in equilibria with constant numbers, but also with periodic fluctuating populations.
\begin{figure}
	\centering
	\includegraphics[width=0.9\linewidth]{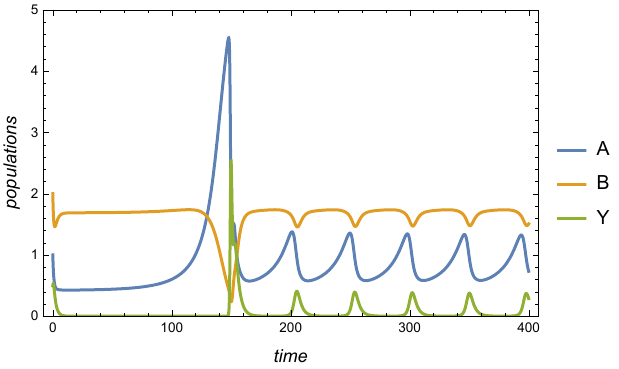}
	\caption{Graph of the populations of the three species as functions of time. After a transient, the solutions become periodic. For the carrying capacity control parameter we are using $K=1.7$ and the remaining parameter values are as in the example of Subsection \ref{sub:ex4}, panel (c) in Fig. \ref{fig:kappaex}.}
	\label{fig:periodic}
\end{figure}

We noticed in Subsection \ref{sub:ex4} that an asymptotically stable ABY equilibrium created at $x=1$ when $K=\kappa(1) \approx 0.80$ became unstable at $K \approx 1.66$ ``out of nothing", or, more precisely, without merging with any other equilibrium. For the same value of $K$, solutions such as the example in Fig. \ref{fig:periodic} appear. Such a phenomenon suggests that hyperbolicity of the stable ABY equilibrium was lost not due to an eigenvalue passing through 0, but due to a pair of complex conjugate eigenvalues crossing the imaginary axis. This phenomenon, in which a stable solution becomes unstable when a pair of conjugate eigenvalues cross the imaginary axis, and \emph{limit cycles} appear, is called a \emph{Hopf bifurcation} \cite{guckhol, perko}. 

We conjecture that what we see in Fig. \ref{fig:periodic} is a solution which orbit converges to a stable limit cycle due to a Hopf bifurcation that occurred for $K \approx 1.66$. By numerical calculation of the eigenvalues, it is possible to see that in fact at $K \approx 1.66$ a conjugate pair of eigenvalues with real part very close to 0 exists. For smaller values of $K$, the real parts of the pair are negative and the ABY equilibrium is asymptotically stable, but real parts become positive for $K$ slightly larger than 1.66 and the equilibrium becomes unstable. 

In the example of Subsection \ref{sub:ex3} the loss of stability of the ABY equilibrium when $K \approx 6.8$ is also numerically seen to be due to a pair of conjugate eigenvalues crossing the imaginary axis. In this case, it seems that we have a Hopf bifurcation, too, but we were not able to find periodic solutions as in Fig. \ref{fig:periodic}. Most probably, the limit cycle created at the Hopf bifurcation is unstable in that example.

In our numerical explorations, we found evidence of Hopf bifurcations and stable periodic solutions with the presence of the three species in other examples. Although we do not have proofs, these observations suggest that stable periodic coexistence of the three species in the BW model may be common and that further work on this topic  is necessary.

\section{Conclusions} \label{sec:conc}
In this paper, we revisited the game-theoretic two-prey-one-predator model introduced by Braga and Wardil \cite{bw}. One of the main concerns of their paper was whether one prey species would lead to extinction of the other, or both prey would coexist stably with the predators.

By rewriting the model as a frequency-dependent Lotka-Volterra-type predator-prey model, we were able to prove several results. In particular, we obtained an equation for calculating equilibrium solutions with the presence of both prey types and predators (ABY equilibria). We provided examples in which, for fixed values of the parameters, we have either no positive ABY equilibrium, one positive ABY equilibrium, or more than one positive ABY equilibrium. By numerical calculation of the eigenvalues of the Jacobian at the equilibria, we showed that ABY equilibria may be stable or unstable. We have also proved very precise theorems on stability or instability of equilibria with only two species present (AY, BY and AB equilibria), with a game-theoretic interpretation for the results. Nonetheless, we were not able to obtain simple conditions for stability or instability of ABY equilibria.

Carrying capacities of the prey are proportional to parameter $K$, see \eqref{BWmodel}, and we found that it is interesting to treat $K$ as a bifurcation parameter. When $K$ is too small, we prove, Theorem \ref{thfamine}, that there are not enough prey to feed predators and, as a consequence, predators are led to extinction. Further increase of $K$ causes changes in the stability status of the many equilibria in the model. In some examples, e.g., the ones illustrated in panels (b) and (d) of Fig. \ref{fig:kappaex}, a rich cascade of bifurcations occurs when $K$ is increased. After overcoming a positivity threshold, AY and BY equilibria may be stable or unstable. According to conditions stated in Theorems \ref{theogtAY} and \ref{theogtBY}, they may or may not change stability status. The AB equilibrium may be stable for small values of $K$ and become unstable for larger values. Positive ABY equilibria may appear and disappear. In some situations, we may have more than one stable equilibrium for the same set of parameters. We also noticed the existence of stable periodic solutions with both prey types and predators.

We proved, see Theorem \ref{thsn}, that pairs of positive ABY equilibria may be created when $K$ overcomes a local minimum value of $\kappa(x)$, or annihilated when $K$ reaches a local maximum value of the same function. As this occurs as a saddle-node bifurcation, we gain information on the dimensions of the stable manifold of the components of the pair. A frequent situation in the examples is one of the ABY equilibria in a pair being stable and the other, unstable.

We also proved, see Theorems \ref{theogtAY}, \ref{theogtBY} and \ref{theoABstab}, that AY, BY and AB equilibria change their stability status whenever they merge with an ABY equilibrium, not necessarily positive. This phenomenon suggests the occurrence of transcritical bifurcations, in which the merging equilibria exchange their stability status. Unfortunately, one condition for the characterization of the bifurcations as transcritical fails, but we observed that the change of stability occurs in all examples we studied. We present the result of these observations as Conjecture \ref{conjtc}. If the conjecture is true, it would be valuable for predicting the stability or instability of ABY equilibria.

Besides saddle-node and similar to transcritical bifurcations, we also found evidence for Hopf bifurcations in some examples, because stable ABY equilibria become unstable without merging with other equilibria at the same value of $K$ in which stable periodic solutions seem to appear. At these values for $K$, we can numerically see that a pair of complex conjugate eigenvalues crosses the imaginary axis.

In conclusion, we may say that the BW model exhibits a very rich gallery of solutions that might explain several ecological situations involving two competing prey types and a common predator. In particular, we notice that changes in the carrying capacities of the prey, due, e.g., to some external factor, may cause abrupt changes such as extinction of one type, or allowance for a rare type to overcome the most frequent one.

We conclude by noticing that we see lots of opportunities for future work on the BW model, e.g., on conditions for the existence of the periodic solutions, or on proving Conjecture \ref{conjtc}.

\appendix
\section{The Lotka-Volterra logistic predator-prey model}\label{app:lvlpp}
Let $P$ denote the size of a prey population and $Y$ the size of a population of predators that feed exclusively on these prey. A variation of the classical Lotka-Volterra predator-prey model \cite{hofsig, edel, murray} is
\begin{equation} \label{lvlpp}
	\begin{cases} 
		P' = (r - \frac{P}{K} - \delta Y)\, P\\
		Y' = ( -\beta + \gamma P) \,Y
	\end{cases}\;.
\end{equation}

All parameters are positive: $r$ denotes the birth rate of prey, $K$ is proportional to the prey carrying capacity, $\delta$ is the capture rate, $\beta$ is the decay rate of the predators if prey are absent, and $\gamma$ measures the efficiency of the predators. If predators are absent, the equation for the prey becomes the well-known \emph{logistic} population model; the \emph{carrying capacity} of the prey is $Kr$.

Fig. \ref{fig:lvlpporb} illustrates orbits of some solutions to \eqref{lvlpp}. The straight lines in the figure are the nullclines $N_1: r - \frac{P}{K} - \delta Y=0$ and $N_2: -\beta + \gamma P=0$ in which either component of the vector field associated with the model vanishes. There are two possible cases according to whether the nullclines intersect in the first quadrant or not. The nullclines divide the first quadrant into regions where each component of the solution to \eqref{lvlpp} either increases or decreases. Above $N_1$ we have $P'<0$ and below $N_1$, $P'>0$. To the left of $N_2$, we have $Y'<0$ and to the right, $Y'>0$.

In the left panel of Fig. \ref{fig:lvlpporb} we show the case $Kr< \frac{\beta}{\gamma}$, in which the nullclines do not intersect in the first quadrant. We prove below that in this case, as illustrated by the orbits, all solutions to the model starting in the first quadrant end up with extinction of the predators, whereas the prey population tends to the carrying capacity. The biological interpretation of this situation is that the limit population $Kr$ of the prey is too small to properly feed the predators.

In the right panel of Fig. \ref{fig:lvlpporb}, we have $Kr> \frac{\beta}{\gamma}$ and the intersection point of the nullclines corresponds to an equilibrium solution of the model. The biological interpretation of this case is that the carrying capacity of the prey is large enough to sustain a population of predators. As illustrated by the orbits, all solutions starting in the first quadrant seem to converge to the equilibrium point. This will also be proved below.

The following result formalizes what is illustrated in Fig. \ref{fig:lvlpporb} and commented above:
\begin{figure}
	\centering
	\includegraphics[width=0.9\linewidth]{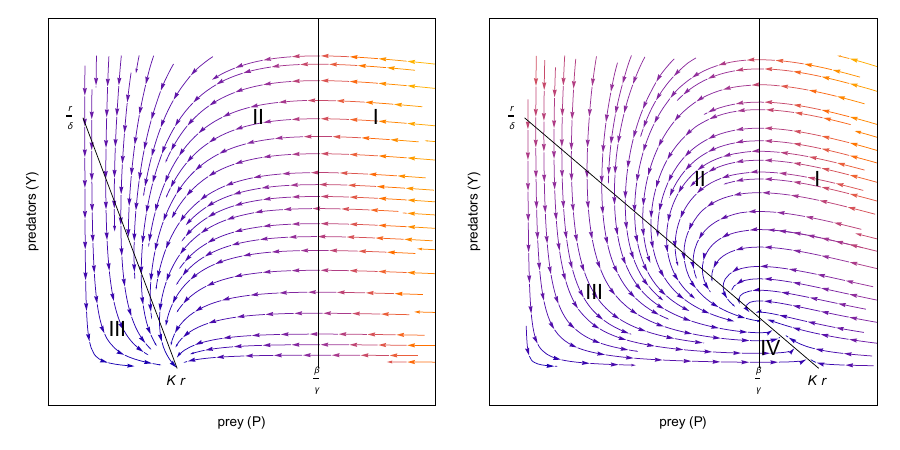}
	\caption{Orbits of the LVLPP model. At left, the $Kr<\frac{\beta}{\gamma}$ case, in which predators are extinct. At right, the case $Kr>\frac{\beta}{\gamma}$, in which predators and prey will coexist stably at an equilibrium point.}
	\label{fig:lvlpporb}
\end{figure}

\begin{theorem}  \label{thlvlpp}
	The maximal solution $(P(t),Y(t))$ of \eqref{lvlpp} with initial condition $(P_0,Y_0)$ in the first quadrant, i.e. $P_0>0$, $Y_0>0$, is defined in an interval that contains $[0, \infty)$.    
	
	Moreover, if $Kr< \frac{\beta}{\gamma}$, $(P(t),Y(t))$ tends to $(Kr,0)$ as $t \rightarrow \infty$. If, on the other hand, $Kr> \frac{\beta}{\gamma}$, $(P(t),Y(t))$ tends to the equilibrium $(\frac{\beta}{\gamma}, \frac{1}{\delta}(r- \frac{\beta}{\gamma K}))$ as $t \rightarrow \infty$.
\end{theorem}
\begin{proof}
	The portions in the first quadrant of the coordinate axes are covered by orbits of the system. Thus, by the uniqueness of solutions to initial value problems, no orbits starting in the first quadrant can leave it.
	
	Consider at first the case $Kr< \frac{\beta}{\gamma}$. The nullclines divide the first quadrant into regions I, II and III as shown in the left panel of Fig. \ref{fig:lvlpporb}. On the boundary between regions I and II the vector field is horizontal pointing to the left. On the boundary between regions II and III the vector field is vertical pointing downwards. So, orbits can pass from region I to II and to region II to III, but cannot traverse in the opposite senses.
	
	An orbit starting in region III cannot leave that region. As $P$ and $Y$ are bounded above and below in III, neither can tend to $\pm \infty$ in finite time. Thus, the maximal solution for an initial condition in region III is defined up to infinite time. The $\omega$-limit set for such an initial condition is well-defined. As in III, $P$ and $Y$ are monotonic and bounded, $\lim_{t \rightarrow \infty}(P(t),Y(t))$ must exist. The existence of this limit, together with the fact that the $\omega$-limit set must be invariant by the flow of the system, implies that it must be an equilibrium point. There are two equilibria in region III, or at its boundary: $(0,0)$ and $(Kr,0)$. Only the second fulfills the requirement that $P$ must increase in III. All orbits starting in III thus tend to $(Kr,0)$.
	
	Orbits starting in region II may either traverse to III or always remain in II. If they traverse, the reasoning in the preceding paragraph applies and the orbits tend to $(Kr,0)$. If they remain in region II, they cannot tend to $\pm \infty$ in finite time, because $P$ is bounded and $Y$ is bounded below by 0 and decreasing. Thus, these solutions are also well-defined to infinite time. As $P$ and $Y$ are both decreasing in region II, orbits that do not leave region II must converge to an equilibrium point. As the only equilibrium in II or at its boundary is $(Kr,0)$, all orbits always in II tend to $(Kr,0)$.
	
	Consider now orbits starting in region I. For them, as $Y$ is increasing in region I, it might be possible that $Y$ tends to infinity in finite time. But, as $P$ decreases, we have, while the orbit lies in I, 
	\[Y'(t) \leq (-\beta+ \gamma P_0)Y\;.\]
	Using Gronwall's inequality \cite{hirschsmaledevaney}, we have 
	\[Y(t) \leq Y_0 e^{(-\beta+ \gamma P_0)t}\;,\]
	which proves that $Y$ cannot tend to infinity in finite time, and, as a consequence, that maximal solutions for initial conditions in region I must be defined up to infinite time. We will prove, however, that such solutions cannot stay in region I for infinite time. In fact, while the solution is in I, we have $P \geq \frac{\beta}{\gamma}$ and $Y \geq 0$. Then 
	\[P'\leq (r- \frac{\beta}{\gamma K})P\;.\]
	Again, by Gronwall's inequality, we must have
	\[P(t) \leq P_0 e^{(r-\frac{\beta}{\gamma K})t}\,,\]
	showing that the solution cannot obey the inequality $P \geq \frac{\beta}{\gamma}$ forever. This proves that orbits starting in I must traverse to II at some time. As they traverse, they must tend, as already shown, to $(Kr,0)$.
	
	To finish the proof, we now consider $Kr > \frac{\beta}{\gamma}$, as illustrated in the right panel of Fig. \ref{fig:lvlpporb}. Considering the signs of $P'$ and $Y'$ in each region, we see that orbits may transition from region I to II, from II to III, from III to IV and from IV to I, but no other transitions are possible. Arguments similar to those in the case $Kr < \frac{\beta}{\gamma}$ prove that maximal solutions are defined up to infinite time, and that no solutions can tend to the equilibria $(0,0)$ and $(Kr,0)$. In principle, orbits could rotate counterclockwise around the equilibrium at the intersection of the nullclines without tending to it. In order to prove that this does not happen, we use, as suggested in \cite{hofsig}, the Lyapunov function
	\[V(P,Y)= \beta \log P- \gamma P+(r- \frac{\beta}{\gamma K}) \log Y- \delta Y\;.\]
	The reader may calculate the time derivative of $V$ along the orbits
	\[ \frac{d}{dt}V(P(t),Y(t)) = \frac{\partial V}{\partial P}P'+ \frac{\partial V}{\partial Y}Y'=\frac{\gamma}{K}(P- \frac{\beta}{\gamma})^2\;.\]
	
	By the Lyapunov theorem \cite{hofsig}, the $\omega$-limit set of any orbit must be contained in the set where the above expression vanishes. But, as $\omega$-limit sets must also be invariant by the flow, in this case all such sets are the equilibrium at the intersection of the nullclines.
\end{proof}

\section{Saddle-node and transcritical bifurcations} \label{app:sntc}
For a review of the basics of the bifurcation theory of dynamical systems, we suggest chapter 3 of \cite{guckhol} or chapter 4 of \cite{perko}. 

Suppose we have a system of $n$ ODEs dependent on a parameter $\mu \in \mathbb{R}$, i.e.
\[z'=f(z,\mu)\;,\]
where $z\in \mathbb{R}^n$ and the vector field $f$ is a smooth function defined in a domain on $\mathbb{R}^n \times \mathbb{R}$ with values in $\mathbb{R}^n$.

Interesting things happen when at a certain value $\mu=\mu_0$ the system has a non-hyperbolic equilibrium at $z=z_0$, $f(z_0,\mu_0)=0$, in which a \emph{single} eigenvalue of the Jacobian matrix $Df(z_0, \mu_0)$ is 0.

Let $v$ be the right eigenvector of $Df(z_0, \mu_0)$ with eigenvalue 0 and $w$ be the corresponding left eigenvector. Suppose also that the number of eigenvalues, counted with their multiplicities, of $Df(z_0, \mu_0)$ with negative real part (the dimension of the stable manifold at the equilibrium $z_0$) is $s$. Of course, the number of eigenvalues of $Df(z_0, \mu_0)$ with positive real part, the dimension of the unstable manifold, is $n-s-1$.

We quote here, as in \cite{guckhol} and \cite{perko}, a theorem by Sotomayor \cite{sotomayor}. Under the additional \emph{transversality} hypotheses
\begin{equation} \label{sn1}
	w \, [\frac{\partial f}{\partial \mu}(z_0,\mu_0)] \neq 0
\end{equation}
and 
\begin{equation} \label{sn2}
	w \, [D^2f(z_0,\mu_0)(v,v)] \neq 0
\end{equation}
there is a smooth curve of equilibria in $\mathbb{R}^n \times \mathbb{R}$ passing through $(z_0,\mu_0)$ and tangent to the hyperplane $\mathbb{R}^n\times \{\mu_0\}$. Depending on the signs of the expressions in \eqref{sn1} and \eqref{sn2}, we may have either no equilibria near $(z_0,\mu_0)$ for $\mu<\mu_0$ and two equilibria near $(z_0,\mu_0)$ for $\mu>\mu_0$, or two equilibria near $(z_0,\mu_0)$ for $\mu<\mu_0$ and no equilibria near $(z_0,\mu_0)$ for $\mu>\mu_0$. The two equilibria close to $(z_0,\mu_0)$ (either for $\mu<\mu_0$ or for $\mu>\mu_0$) are hyperbolic and one of them has a stable manifold of dimension $s$, and the other has a stable manifold of dimension $s+1$. 

This phenomenon is called a \emph{saddle-node} bifurcation at $\mu=\mu_0$. Notice that as the bifurcation parameter $\mu$ is increased, we may have at $\mu=\mu_0$ either the creation of a \emph{pair} of equilibria, or the annihilation of a \emph{pair} of equilibria.

Sotomayor \cite{sotomayor} also proved that if instead of \eqref{sn1} and \eqref{sn2} we have
\begin{equation} \label{tc1}
	w \, [\frac{\partial f}{\partial \mu}(z_0,\mu_0)] = 0\;,
\end{equation}
\begin{equation} \label{tc2}
	w \, [D \frac{\partial f}{\partial \mu}(z_0, \mu_0) v] \neq 0
\end{equation}
and again \eqref{sn2}, then the phenomenon is called a \emph{transcritical} bifurcation. In this case, for each $\mu$ in a neighborhood of $\mu_0$ we have two equilibria. For $\mu<\mu_0$ one equilibrium has a stable manifold with dimension $s$, the other of $s+1$ as dimension for the stable manifold. For $\mu>\mu_0$ the two equilibria \emph{exchange stability}, i.e. the one with stable manifold with dimension $s$ increases to dimension $s+1$, the other decreases from $s+1$ to $s$. Contrary to the case of the saddle-node bifurcation, here we always have two equilibria in a neighborhood of $\mu_0$, which merge when $\mu=\mu_0$ and exchange their stabilities.

\section*{Acknowledgments}
The first author was supported by a CNPq-Brasil scholarship, the second author is supported by FAPEMIG grants APQ-01784-22 and RED-00133-21.

The authors thank Artur C. Fassoni and Denis C. Braga for helpful suggestions and discussion.

\end{document}